\documentclass[nologo,url,11pt,a4paper]{ETHpaper} 

\usepackage{amsmath}
\usepackage{amsthm}
\usepackage{amssymb}
\usepackage{url}
\usepackage{bbm}
\usepackage{epsfig}
\usepackage{subfig}
\usepackage{caption}

\newtheorem{thm}{Theorem}
\newtheorem{prop}{Proposition}
\newtheorem{cor}{Corollary}

\begin{document}
\title{An ensemble perspective on multi-layer networks}
\reference{To appear in: A. Garas (Ed.) {\it Interconnected Multilayer Networks}, Springer 2015 }
\author{Nicolas Wider, Antonios Garas, Ingo Scholtes, Frank Schweitzer}
\address{Chair of Systems Design, ETH Zurich, Switzerland\\
  \url{www.sg.ethz.ch}}
\www{\url{http://www.sg.ethz.ch}}
\makeframing
\maketitle

\begin{abstract}
We study properties of multi-layered, interconnected networks from an ensemble perspective, i.e. we analyze ensembles of multi-layer networks that share similar aggregate characteristics.
Using a diffusive process that evolves on a multi-layer network, we analyze how the speed of diffusion depends on the aggregate characteristics of both intra- and inter-layer connectivity.
Through a block-matrix model representing the distinct layers, we construct transition matrices of random walkers on multi-layer networks, and estimate expected properties of multi-layer networks using a mean-field approach.
In addition, we quantify and explore conditions on the link topology that allow to estimate the ensemble average by only considering aggregate statistics of the layers.
Our approach can be used when only partial information is available, like it is usually the case for real-world multi-layer complex systems.
\end{abstract}

\section{Introduction}\label{ch2.5-intro}

Networks are often used to describe interactions among the elements of a complex system.
But until recently, the standard assumption in the literature was that networks are isolated entities and do not interact with other networks.
Today we understand that this assumption is a rough simplification, since real networks usually have complex patterns of interaction with other networks.
In order to study more realistic systems, network theory extended its perspective to account for these network to network interactions, and to investigate their influence on various processes of interest that may use the network topology as substrate~\cite{ch2.5-Parshani2010,ch2.5-Gao2011,ch2.5-Gomez2013,ch2.5-Garas2014}.

Networks consisting of multiple networks and the connections between them are called  \textit{interconnected} or \textit{multi-layer} networks~\cite{ch2.5-boccaletti2014}.
In a multi-layered network each individual layer contains a network that is different from the networks contained in other layers, and the layer interconnectivity refers to the fact that nodes in different layers can be connected to each other.
Nevertheless, it is often possible to extend and apply methods developed for single-layer (isolated) networks to multi-layer networks, assuming that all layers and the connections between them are known precisely.
Unfortunately, when creating networks using relational data on real-world systems we are often confronted with situations where we \emph{lack information} about the details of their multi-layer structure.
In such situations, ensemble-based approaches allow us to reason about the expected properties of such networks, provided that we have access to aggregate statistics which can be used to define a statistical ensemble.

For instance, there are situations in which we are able to precisely map the topology \emph{within} each layer individually, but we may not be able to obtain the detailed topology of connections \emph{across} different layers.
As an example, we may consider the topology of connections between users in different online social networks (OSNs).
Such a system can be represented as a multi-layer network, where each layer represents the network of connections between users \emph{within} one OSN.
In addition, cross-layer connections are due to users which are members of multiple OSNs at the same time, and which can thus drive the dissemination of information across OSNs.
Data on the network topology within particular OSNs are often readily available, however it is in general very difficult to identify accounts of the same user in different OSNs.

Contrary to the situation described above, we may also consider situations in which detailed information on the topology of cross-layer links is available, while the detailed topology of connections \emph{within} layers is not known.
For example, there may be a rather small number of static links \emph{across} layers, while the topology of links \emph{within} layers is too large and too dynamic to allow for a detailed mapping.
Again, in such a situation we may still have access to partial, aggregate information on the \emph{inter-layer} connectivity (such as the number of nodes or the density of links) which we can use in order to reason about a multi-layer complex system.

Both the above situations lead to multi-layer networks, and both require us to reason about a system with incomplete information.
This problem can be addressed from a macroscopic perspective using \emph{statistical ensembles}, and in this work we extend the ensemble perspective to multi-layer networks, where we have access to mere aggregate statistics either on links \emph{within} or \emph{across} layers.
Combining both detailed and aggregate information on the links in a multi-layer network, we first define a statistical ensemble, i.e. a probability space containing all network realizations that are consistent with available information.
Secondly, we assume a probability mass function which assigns a probability to each possible realization in the ensemble.
And finally, using either analytical or numerical techniques, we use the resulting probability space to reason about the \emph{expected properties} of a network given that it is drawn from the ensemble.

The rest of the manuscript is structured as follows.
In Sect.~\ref{methods} we present our methodological approach to model ensembles of multi-layer networks, we formally introduce the diffusion process that is assumed to run on the multi-layer network, and we introduce a method that allows to aggregate the statistics of links inside layers and across layers.
In Sect.~\ref{meanfield} we introduce a mean-field approach to approximate  ensemble averages, and we investigate under which conditions it can be used to argue about diffusion in multi-layer networks.
In particular we discuss three distinct cases according to different levels of information that we may have about the topology of links across the layers or inside the layers.

\section{Methods and definitions}\label{methods}

In our analysis we investigate a diffusion process that evolves on a static multi-layer network.
More precisely we focus on diffusion dynamics modeled by a random walk process.
In the following we provide a short summary that highlights the important properties of this process.
Note that these are facts already known from the theory of random walks on single layer networks~\cite{ch2.5-lovasz1993,ch2.5-Blanchard2011}, but later on we will extend this framework to multi-layer networks.

We assume a discrete time random walk process on a network $\mathbf{G}$ that consist of $n$ nodes.
Starting at an arbitrary node, at each step of the process the walker moves to an adjacent node, so that for a pair of nodes $i,j$ the probability for a walker to move from node $i$ to node $j$ is given by the corresponding entry $\mathbf{T}_{ij}$ of a \emph{transition matrix} $\mathbf{T}$.
Since we have $\sum_j P(i\rightarrow j)=1$, the transition matrix is row stochastic.

We further consider a vector $\pi^t \in \mathbf{R}^n$, whose entries $\pi_i^t$ indicate the probability of a random walker to visit node $i$ after $t$ steps of the process.
Here, we consider $\pi^0$ as a given initial distribution, whose entries $\pi_i^0$ give the probability that the random walker has started at node $i$.
The change of visitation probabilities $\pi^t \rightarrow  \pi^{t+1}$ can then be calculated based on the transition matrix as follows:
\begin{equation}
 \pi^{t+1} = \pi^t\mathbf{T}.
\end{equation}
Since this is an iterative process starting with $\pi_0$, the visitation probability vector after $t$ time steps can be calculated as $\pi^t = \pi^0\mathbf{T}^t$, and we can investigate the long-term behavior of the random walk process for $t\rightarrow\infty$.
For a visitation probability vector $\pi^*$ such that $\pi^*\mathbf{T}=v^*$, we can say that the process reaches a stationary distribution $\pi^*$, and if the transition matrix $\mathbf{T}$ is primitive, the Perron-Frobenius theorem guarantees that such a unique stationary distribution $\pi^*$ exist.

In order to assess the convergence time of a random walk process, we can study the \emph{total variation distance} between visitation probabilities $\pi^t$ after $t$ steps and the stationary distribution $\pi^{*}$.
For two distributions $\pi$ and $\pi'$, the total variation distance is defined according to Ref.~\cite{ch2.5-Rosenthal1995} as
\begin{equation}
 \Delta(\pi,\pi'):=\frac{1}{2}\sum_i |\pi_i-\pi'_i|\, ,
\end{equation}
where $\pi_i$ indicates the $i$-th entry of $\pi$.

As a proxy for diffusion speed, we can now investigate how long it takes until the total variation distance $\Delta(\pi^t, \pi^{*})$ falls below some given threshold value $\epsilon$ (for a small $\epsilon$).
In other words, we study how many steps $t(\epsilon)$ a random walker needs such that $\Delta(\pi^{t},\pi^{*})\leq\epsilon$ for $t \geq t(\epsilon)$

The eigenvalues $1 = \lambda_1 \geq |\lambda_2| \geq \ldots \geq |\lambda_n|$ of a row-stochastic matrix necessarily have absolute values that fall between zero and one, while the largest eigenvalue $\lambda_1$ is necessarily one.
The number of required time steps $t(\epsilon)$ (and thus the diffusion speed of the random walk process) can be estimated by means of the the second-largest eigenvalue $\lambda_2$ of $\mathbf{T}$,
\begin{equation}
 t(\epsilon)\sim\frac{-1}{\ln(|\lambda_2|)}\,.
 \label{eq:lambda2}
\end{equation}
For a detailed derivation see Ref.~\cite{ch2.5-Chung2005}.
Eq.~(\ref{eq:lambda2}) shows that a second-largest eigenvalue $\lambda_2$ close to one implies slow convergence, while $\lambda_2$ close to zero implies fast convergence.
Therefore in the following we use the second-largest eigenvalue of a transition matrix $\lambda_2(\mathbf{T})$ as a proxy to measure and quantify the convergence behavior on a network.

\subsection*{Multi-layer network}\label{multi}
The purpose of our study is to investigate diffusion processes on ensembles of networks with multiple interconnected layers.
Thus, in the following we briefly introduce the notion of multi-layer networks used in this manuscript.
Let us consider a multi-layer network denoted by $\mathbf{G}$ that consist of $L$ non-overlapping layers $G_1,\ldots,G_L$.
Each of these layers $G_l$ is a single-layer network $G_l=(V_l,E_l)$ where $V(G_l)$ and $E(G_l)$ denote the nodes and links of layer $l$ respectively.
We call the links $E(G_l)$ between nodes \emph{within} the layers $l$ \textit{intra}-links.
The multi-layer network $\mathbf{G}$ consist in total of $n$ nodes, where $n=\sum_{l=1}^L |V(G_l)|$.
In addition, we assume a set $E_I(\mathbf{G})$ of \emph{inter-layer} links which connect nodes \emph{across} layers, i.e.
for each edge $(u,v) \in E_I$ we have $u \in V(G_i)$ and $V(G_j)$ for $i \neq j$.
Inter-layer links induce a multipartite network with the independent sets $G_1,\ldots,G_L$.

In our study we consider undirected and unweighted networks, however some of our results may hold even for directed or weighted networks.
Furthermore, from the perspective of a random walk process, we assume that inter- and intra-layer links are indistinguishable, i.e. transitions are made purely randomly irrespective of the type of link.
As such, the multi-layer network can also be viewed as a huge single network consisting of subgraphs $G_1, \ldots, G_L$.

As mentioned above diffusion dynamics on networks can be studied analytically using transition matrices of random walkers~\cite{ch2.5-lovasz1993,ch2.5-Domenico2013}.
The multi-layer structure of a network can explicitly be incorporated in a random walk model by constructing a so-called \emph{supra}-transition matrix~\cite{ch2.5-Gomez2013,ch2.5-Ribalta2013} similar to the supra-adjacency matrix used in ~\cite{ch2.5-Domenico2013,ch2.5-Multilayer2013,ch2.5-DeDomenico2013}.
The supra-adjacency matrix of a multi-layer network $\mathbf{G}$ can be defined in a block-matrix structure as
\begin{equation}
  \mathbf{A} =
  \left(
    \begin{array}{c|c|c|c|c}
      \mathbf{A}_1 & \ldots & \mathbf{A}_{1t}& \ldots& \mathbf{A}_{sL} \\
      \hline
      \vdots & \ddots & \vdots& \ddots& \vdots \\
      \hline
      \mathbf{A}_{s1} & \ldots & \mathbf{A}_{st}&\ldots&  \mathbf{A}_{sL} \\
      \hline
      \vdots & \ddots & \vdots& \ddots& \vdots \\
      \hline
      \mathbf{A}_{L1} & \ldots & \mathbf{A}_{Lt}&\ldots&  \mathbf{A}_L
    \end{array}
   \right)\,.
   \label{eqn:NoN}
\end{equation}
On the diagonal we have the adjacency matrices $\mathbf{A}_1,\ldots,\mathbf{A}_L$ corresponding to the layers $G_1,\ldots,G_L$, thus entries of these block matrices  represent the \emph{intra-layer links} of the multi-layer network.
Off-diagonal matrices $\mathbf{A}_{ij}$ for $i,j\in\{1,\ldots,L\}$ with $i \neq j$ represent \emph{inter-layer} links that connect nodes in layer $G_i$ to nodes in layer $G_j$.
Since we consider undirected networks we have $\mathbf{A}_{ij}^\top=\mathbf{A}_{ji}$.

Based on a supra-adjacency matrix $\mathbf{A}$ we can easily define a \emph{supra-transition} matrix $\mathbf{T}$ of a random walker on a multi-layer network $\mathbf{G}$. In block-matrix form such a matrix can be written as:
\begin{equation}
  \mathbf{T} =
  \left(
    \begin{array}{c|c|c|c|c}
      \mathbf{T}_1 & \ldots & \mathbf{T}_{1t}& \ldots& \mathbf{T}_{sL} \\
      \hline
      \vdots & \ddots & \vdots& \ddots& \vdots \\
      \hline
      \mathbf{T}_{s1} & \ldots & \mathbf{T}_{st}&\ldots&  \mathbf{T}_{sl} \\
      \hline
      \vdots & \ddots & \vdots& \ddots& \vdots \\
      \hline
      \mathbf{T}_{L1} & \ldots & \mathbf{T}_{Lt}&\ldots&  \mathbf{T}_L
    \end{array}
   \right)\,.
   \label{eqn:non2}
\end{equation}
Here, each entry $T_{ij}$ is defined as:
\begin{equation}
 T_{ij}=\frac{a_{ij}}{\sum_{k=1}^n a_{kj}}\, ,
\end{equation}
where $a_{ij}$ are the corresponding entries of the supra-adjacency matrix $\mathbf{A}$.
Note that, due the presence of inter-layer links, block matrices $\mathbf{T}_{ij}$ are in general not equal to the row-normalized version of block matrices $\mathbf{A}_{ij}$.
The supra transition matrix defined above can be used to model a random walk process on a multi-layer network.

From an analytical perspective the supra-transition matrix can be treated in the same way as the transition matrix of a single layer as explained above.
In the case of undirected networks the eigenvalues of a transition matrix are related to the eigenvalues of the normalized Laplacian matrix.
In our case we study the second-largest eigenvalue of the supra-transition matrix and use it as a proxy for the efficiency of a network with respect to a diffusion process as pointed out above.

Using $\mathbf{T}$ we are able to model a diffusion process on a multi-layer network.
Since we especially want to emphasize the relevance of the inter-links, in the next section we introduce a transition matrix that only considers transitions across layers and not between individual nodes.
As we will see later, this aggregated transition matrix is useful to distinguish the influence of inter-layer and intra-layer links on the convergence behavior of a random walk process.

\subsection*{Multi-layer aggregation}\label{multiagg}
The supra-transition matrix $\mathbf{T}$ introduced previously contains transition probabilities for any pair of nodes in the multi-layer network.
In this sense $\mathbf{T}$ could also be the transition matrix of a large network, which is not divided in separate layers.
In order to understand the effects of a layered structure, in this section we focus explicitly on transitions across layers.
To do this we aggregate the statistics of inter-links and the intra-links of all single layers, and thus, we homogenize all individual nodes that belong to the same layer.
This way we reduce the supra-transition matrix $\mathbf{T}$ of dimension $n$ to an aggregated transition matrix $\mathfrak{T}$ of dimension $L$.
We call this process \emph{multi-layer aggregation} and the matrix $\mathfrak{T}$ the \emph{layer-aggregated} or just \emph{aggregated} transition matrix.
Later on we will provide a relation between the eigenvalues of $\mathbf{T}$ and $\mathfrak{T}$, which will allow us to decompose the spectrum of $\mathbf{T}$.
This is important since the convergence behavior of a random walk process depends on the second largest eigenvalues of $\mathbf{T}$.

Let us begin by discussing the construction process of the layer-aggregated transition matrix.
Our goal is to define transition probabilities across any two layers $G_s$  and $G_t$ by averaging the transitions between any two nodes of $G_s$ and $G_t$.
Under certain conditions which will be specified in the following, these average transition probabilities can be representative for all nodes of the different layers.

Let $\mathbf{G}$ be a multi-layer network that consists of $L$ layers $G_1,\ldots,G_L$.
The transition probability to go from node $v_i$ to any node $v_j$ in $\mathbf{G}$ is defined as
\begin{equation}
 P(v_i\rightarrow v_j)=\frac{\omega(v_i,v_j)}{\sum_{k}\omega(v_i,v_k)}
\end{equation}
where $\omega(v_i,v_j)$ is the weight of a link connecting $v_i$ with $v_j$.
This is a general formalism, but since we only consider unweighted networks we have $\omega(v_i,v_j)=1$ if and only if there is a link between the nodes $v_i$ and $v_j$.

For each node $v_i$ in layer $G_s$ we require that the transition probabilities $P(v_i\rightarrow *)$ to nodes in another layer $G_t$ fulfill the following equation
\begin{equation}
  \alpha_{ss}\sum_{v_j\in V(G_s)}P(v_i\rightarrow v_j)=\alpha_{st}\sum_{v_k\in V(G_t)}P(v_i\rightarrow v_k)\quad \forall v_i\in V(G_s)\, ,
 \label{eqn:cond}
\end{equation}
where $\alpha_{st}$ is a factor that only depends on the layers $G_s$ and $G_t$.
The factor $\alpha_{ss}$ is used to normalize the transitions, such that $\sum_t\alpha_{st}=1$ is satisfied.
In other words Eq.(\ref{eqn:cond}) implies that the probability for a random walker at node $i$ to stay inside layer $G_s$ is a multiple of the probability to switch to layer $G_t$.

We can see that $\alpha_{st}$ is independent of $i$, and therefore $\mathbf{T}_{st}=\alpha_{st}\mathbf{R}_{st}$ where $\mathbf{R}_{st}$ is a row stochastic matrix.
This means that $\mathbf{T}_{st}$ resembles a scaled transition matrix, and $\alpha_{st}$ represents the weighted fraction of all links starting in $G_s$ that end up in $G_t$.
Thus, we can define the aggregation of a supra-adjacency matrix satisfying Eq.(\ref{eqn:cond}) as
\begin{equation}
 \mathfrak{T} =
  \left(
    \begin{array}{c|c|c|c|c}
      \alpha_{11} & \ldots & \alpha_{1t}& \ldots& \alpha_{sL} \\
      \hline
      \vdots & \ddots & \vdots& \ddots& \vdots \\
      \hline
      \alpha_{s1} & \ldots & \alpha_{st}&\ldots&  \alpha_{sL} \\
      \hline
      \vdots & \ddots & \vdots& \ddots& \vdots \\
      \hline
      \alpha_{L1} & \ldots & \alpha_{Lt}&\ldots&  \alpha_{LL}
    \end{array}
   \right)\,.
\end{equation}

If a multi-layer network $\mathbf{G}$ satisfies Eq.(\ref{eqn:cond}) 
we can follow that the spectrum of the aggregated matrix $\mathfrak{T}=\{\alpha_{st}\}_{st}$ is
\begin{equation}
 Spec(\mathfrak{T})=\{1,\lambda_2,\ldots,\lambda_L\}\, ,
\end{equation}
and it holds that $\lambda_2,\ldots,\lambda_L\in Spec(\mathbf{T})$ (see Prop.~\ref{prop:rev} in the Appendix).

This relation implies that the aggregated matrix $\mathfrak{T}$ preserves $L$ eigenvalues of the supra-transition matrix $\mathbf{T}$, where $L$ is the amount of layers.
In other words, under the condition that Eq.(\ref{eqn:cond}) holds, we are able to make statements about the spectrum of the transition matrix $\mathbf{T}$ only using the layer-aggregated transition matrix $\mathfrak{T}$.

Similar to the Fiedler vector, i.e. the eigenvector corresponding to the second smallest eigenvalue of the Laplacian matrix, here we may use the eigenvector $v_2$ corresponding to the second largest eigenvalue $\lambda_2$ of the transition matrix $\mathbf{T}$.
The vector $v_2$ contains negative and positive entries and sums up to zero.
If all individual nodes that belong to the same layer correspond to entries of $v_2$ with the same sign, we consider the layers of $\mathbf{G}$ partitioned according to $v_2$, which is also called spectral partitioning or spectral bisection~\cite{ch2.5-Fiedler1973,ch2.5-Fiedler1975}.
In this case, according to Cor.~\ref{cor:eig} in the Appendix, it holds that $\lambda_2(T)=\lambda_2(\mathfrak{T})$.

We note that the multi-layer aggregation, performed according to a spectral partitioning, has similarities to spectral coarse-graining~\cite{ch2.5-Gfeller2007}.
The multi-layer aggregation presented here decreases the state space as well, but still preserves parts of the spectrum.

The spectral properties introduced in this section are important for our ensemble estimations that follows, since we characterize the diffusion process by its convergence efficiency measured through the second-largest eigenvalue $\lambda_2(\mathbf{T})$ of the supra-transition matrix.
However, as outlined before, if Eq.(\ref{eqn:cond}) holds then this eigenvalue is equal to the second-largest eigenvalue $\lambda_2(\mathfrak{T})$ of the aggregated transition matrix $\mathfrak{T}$.
Considering that for the construction of $\mathfrak{T}$ we only used aggregated statistics on the network and not the detailed topologies of the inter-links or any of the intra-links of all single layers, this already provides a hint how we can treat a system in the case of limited information.

\section{Mean-field approximation of ensemble properties}\label{meanfield}
With the layer aggregation introduced in the previous section, we are now able to deal with multi-layer network ensembles in case of limited information.
In our case, this information concerns knowledge either of the inter-link topology between layers or the intra-link topologies of all single layers.
For our purpose we define ensembles based on the inter-link densities and intra-link densities of all single layers,
more precisely, by using the amount of nodes, the amount of inter-links across any two layers, and the amount of intra-links of all single layers.
The number of nodes in individual layers are represented by the vector $\vec{n}=\{n_1,\ldots,n_L\}$ and the number of links between layers by a matrix $\mathbf{M}$ with entries $m_{st}$ where $s$ gives the source layer and $t$ the target layer.
Intra-layer links have both of their ends in the same layer and therefore we assume that the diagonal elements $m_{ss}$ are equal to the amount of desired intra-links multiplied by two.
We denote the ensemble defined by these two quantities $\mathcal{E}(\vec{n},\mathbf{M})$.

A single random realization of this ensemble satisfies the aggregated statistics given by $\mathbf{M}$ and $\vec{n}$.
We assume a random uniform distribution of links and therefore each realization of $\mathcal{E}(\vec{n},\mathbf{M})$ has the same probability.
However, instead of single realizations we are rather interested in the average values of all possible realizations.
For each multi-layer network realization $\mathbf{G}$ of $\mathcal{E}(\vec{n},\mathbf{M})$ we build the supra-transition matrix $\mathbf{T}$, which defines a random walk process that is different for every realization.
As discussed above, a proxy of the convergence quality of these random walk processes is given by the second-largest eigenvalue $\lambda_2(\mathbf{T})$.
Our goal is to estimate the average $\lambda_2$ of the ensemble $\mathcal{E}(\vec{n},\mathbf{M})$, and we do this using a mean-field approach on the supra-transition matrix $\mathbf{T}$ that is similar to Refs.~\cite{ch2.5-Grabow2012,ch2.5-Martin2014}.

Hereafter we will provide a mean-field approach for the general case, i.e. when the exact topology of inter-links and intra-links of all single layers are unknown.
Next, building on this approach, we will discuss the case for which we have full knowledge of the intra-link topology but we do not know the inter-link topology, and the case for which we have full knowledge of the inter-link topology but we do not know the intra-link topology.

\subsection*{Case I: unknown topology of inter-links and intra-links for all layers}

For this case we only assume knowledge of the ensemble parameters $\mathbf{M}$ and $\vec{n}$.
We define a mean-field adjacency matrix $\hat{\mathbf{A}}$ with a block structure similar to Eq.(\ref{eqn:NoN}), and for each $\mathbf{\hat{A}}_{st}$ we are only given the amount of links equal to $m_{st}$.
Since we do not know how these links are assigned to the entries $\mathbf{A}_{st}$, without loss of generality we assume a uniform distribution.
Thus, for the blocks of $\hat{\mathbf{A}}$ we have
\begin{equation}
 \mathbf{\hat{A}}_{st}=\left\{\frac{m_{st}}{n_s n_t}\right\}_{ij},\quad i\in \{1,\ldots,n_s\},\quad j\in\{1,\ldots,n_t\}\,.
\end{equation}

Following the discussion of Sect.~\ref{methods}, based on the mean-field adjacency matrix we define a mean-field transition matrix $\hat{\mathbf{T}}$.
The transition probability between any two nodes $i,j\in G_s$ for a fixed layer $s$ is the same since according to the available information individual nodes cannot be distinguished based on their connectivity.
Further, the transition probabilities between any two nodes $i\in G_s$ and $j\in G_t$ are the same for any two fixed layers $s$ and $t$.
Therefore, all block transition matrices $\mathbf{\hat{T}}_{st}$ contain the same value at each entry.
Hence we have
\begin{equation}
 \mathbf{\hat{T}}_{st}=\left\{\frac{m_{st}}{n_t\left(\sum_{k} m_{sk}\right)}\right\}_{ij},\quad i\in \{1,\ldots,n_s\},\quad j\in\{1,\ldots,n_t\}\,.
\end{equation}
Now, using Eq.(\ref{eqn:cond}) we can construct an aggregated supra-transition matrix $\mathfrak{T}$ with entries
\begin{equation}
\alpha_{st}=\frac{m_{st}}{\sum_{k} m_{sk}}\,.
\end{equation}
The aggregated supra-transition matrix $\mathfrak{T}$ describes the macro behavior of the multi-layer network ignoring the detailed topology of the inter-links and the intra-links of all single layers.
Since $\mathfrak{T}$ depends on a mean-field approach it only captures probabilistic assumptions of the ensemble $\mathcal{E}(\vec{n},\mathbf{M})$.
Thus, the spectrum of the mean-field supra-transition matrix $\hat{\mathbf{T}}$ can be calculated by
\begin{equation}
 Spec(\hat{\mathbf{T}})=Spec(\mathfrak{T})\cup\left(\bigcup_{s=1}^{L}\cup_{i=1}^{n_s-1}\{0\}\right)\,.
\end{equation}

To clarify the situation, let us briefly discuss the simple case of a network $\mathbf{G}$ that contains only two layers $G_1$ and $G_2$, for which we get
\begin{equation}
 \mathfrak{T}=
  \left(
    \begin{array}{cc}
    1-\alpha_{12}&\alpha_{12}\\
    \alpha_{21}&1-\alpha_{21}
    \end{array}
   \right)\,.
\end{equation}
Hence, for the mean-field matrix of a two-layered network we obtain
\begin{equation}
 Spec(\hat{\mathbf{T}})=\{1,1-\alpha_{12}-\alpha_{21},\underbrace{0,\ldots,0}_{|n|-2\text{ times}}\}\, .
\end{equation}
These results are remarkable, since the layer-aggregated transition matrix captures the same relevant eigenvalues as the mean-field transition matrix.
So, for the case of a diffusion process in two layers the eigenvalue of interest is $\lambda_2(\hat{\mathbf{T}})=1-\alpha_{12}-\alpha_{21}$.
However, so far we only considered the general case where we can only use the densities of inter-links and intra-links of all single layers.
In the following two sections we will investigate cases where we may have some additional information about either the inter-link topology between all single layers or the intra-link topology of all single layers.
For simplicity, we will restrict ourselves to the two layer case but, as shown in the appendix, our results can be generalized to multiple layers.

\subsection*{Case II: unknown inter-connectivity}\label{inter}

For this case we assume full knowledge of the intra-link topology, i.e. we know exactly which nodes are connected in all of the single layers.
But while we know the number of links between the layers we do not know how the layers are connected, i.e. we do not know the inter-link topology.
With respect to the general case discussed previously, 
here we have more information which is expected to improve the estimates of the ensemble average.

More precisely, we consider a two-layer network with unknown inter-link structure denoted by $E_I(\mathbf{G})$, but with a given amount of $m$ interconnecting links which connect the networks $G_1$ and $G_2$.
This means that the diagonal blocks $\mathbf{A}_1$ and $\mathbf{A}_2$ of the supra-adjacency matrix are given, but the off-diagonal blocks $\mathbf{A}_{12}$ and $\mathbf{A}_{21}$ can take any form such that they have exactly $m$ entries different from zero.
Since there are no further constraints on the ensemble, any random link configuration that consists of $m$ inter-links has the same probability to occur.
Therefore, we define the mean-field supra-adjacency blocks that correspond to the inter-links, $\mathbf{\hat{A}}_{12}$ and $\mathbf{\hat{A}}_{21}$, to have the same value $\frac{m}{n_1 n_2}$ in each entry.

For the supra-transition matrix we have to row-normalize $\mathbf{A}_1$ with $\mathbf{\hat{A}}_{12}$ and $\mathbf{\hat{A}}_{21}$ with $\mathbf{A}_{2}$.
The row sums of $\mathbf{\hat{A}}_{12}$ are all equal to $m/n_1$ and the row sums of $\mathbf{\hat{A}}_{21}$ are all equal to $m/n_2$, while the row sums of $\mathbf{A}_1$ and $\mathbf{A}_2$ correspond to the individual degrees of the nodes in $G_1$ and $G_2$ respectively.
Thus, we use the mean degree $\hat{d_1}$ of $G_1$ and $\hat{d_2}$ of $G_2$ in order to obtain the row-normalized transition matrix $\hat{\mathbf{T}}$, and to define the following factors
\begin{equation}\label{eq:fac}
 \alpha_{1}=\frac{n_1\hat{d_1}}{n_1\hat{d_1}+m},\quad \alpha_{2}=\frac{n_2\hat{d_2}}{n_2\hat{d_2}+m},\quad \alpha_{12}=\frac{m}{n_1\hat{d_1}+m},\quad \alpha_{21}=\frac{m}{n_2\hat{d_2}+m}\,.
\end{equation}
Note that $\alpha_1+\alpha_{12}=1$ and $\alpha_2+\alpha_{21}=1$.

Accordingly we define the mean transition blocks of $\mathbf{T}_{12}$ and $\mathbf{T}_{21}$.
\begin{align}
 \mathbf{\hat{T}}_{12}&=\left\{\frac{m}{n_2(n_1\hat{d_1}+m)}\right\}_{ij}\quad\text{for}\quad i\in\{1,\ldots,n_1\},j\in\{1,\ldots,n_2\}\\
 \mathbf{\hat{T}}_{21}&=\left\{\frac{m}{n_1(n_2\hat{d_2}+m)}\right\}_{ij}\quad\text{for}\quad i\in\{1,\ldots,n_2\},j\in\{1,\ldots,n_1\}\, .
\end{align}
This means that each of the off-diagonal block matrices that correspond to the mean-field inter-link structures have the same value at each matrix element, and the diagonal blocks are just rescaled transition matrices of $\mathbf{A}_1$ and $\mathbf{A}_2$,
\begin{equation}
 \mathbf{\hat{T}}_{1}=\left(1-\alpha_{12}\right)T(\mathbf{A}_1),\quad \mathbf{\hat{T}}_{2}=\left(1-\alpha_{21}\right)T(\mathbf{A}_2)\,,
\end{equation}
where $T(\mathbf{M})$ is the row-normalized version of matrix $\mathbf{M}$.
We denote the supra-transition matrix with the blocks constructed as described before by $\hat{\mathbf{T}}$,
\begin{equation}
 \hat{\mathbf{T}}=
 \left(
    \begin{array}{c|c}
      \mathbf{\hat{T}}_{1} & \mathbf{\hat{T}}_{12} \\
      \hline
      \mathbf{\hat{T}}_{21} & \mathbf{\hat{T}}_{2}
    \end{array}
   \right)\,.
\end{equation}

This mean-field matrix has some special properties.
First of all, the eigenvalues of $\mathbf{\hat{T}}_{1}$ and $\mathbf{\hat{T}}_{2}$ are also eigenvalues of $\hat{\mathbf{T}}$.
Further, the multi-layer aggregation of $\hat{\mathbf{T}}$ is given by
\begin{equation}
 \mathfrak{T}=
   \left(
    \begin{array}{cc}
      \alpha_{1} & \alpha_{12} \\
      \alpha_{21} & \alpha_{2}
    \end{array}
   \right)=
  \left(
    \begin{array}{cc}
      1-\alpha_{12} & \alpha_{12} \\
      \alpha_{21} & 1-\alpha_{21}
    \end{array}
   \right)\, ,
\end{equation}
so, the second-largest eigenvalue of $\mathfrak{T}$ is given by $\lambda_2=1-\alpha_{12}-\alpha_{21}$.

The second-largest eigenvalues of $\hat{\mathbf{T}_1}$ is equal to $(1-\alpha_{12})\lambda^1_2$ and of $\hat{\mathbf{T}_2}$ is equal to $(1-\alpha_{21})\lambda^2_2$, where $\lambda_2^1=\lambda_2(T(\mathbf{A}_1))$ and $\lambda_2^2=\lambda_2(T(\mathbf{A}_2))$.
Therefore the second largest eigenvalue of $\hat{\mathbf{T}}$, denoted by $\lambda_2(\hat{\mathbf{T}})$, fulfills the following condition (See Prop.~\ref{prop:spec} in the Appendix for more details)
\begin{equation}
 \lambda_2(\hat{\mathbf{T}})=\max\left(1-\alpha_{12}-\alpha_{21},(1-\alpha_{12})\lambda^1_2,(1-\alpha_{21})\lambda^2_2\right)\,.
 \label{eq:maxterm1}
\end{equation}

We would like to remind the reader that an eigenvalue $\lambda_2$ close to one implies slow convergence and $\lambda_2$ close to zero fast convergence.
From the above equation we can see that as long as $\lambda_2=1-\alpha_{12}-\alpha_{21}$ is maximal the inter-links are the limiting factor of the convergence in the multi-layer network.
This means that due to the inter-link topology the random walk diffusion is slowed down, and the influence of the intra-layer topologies is marginal to the process.

When either the term of $\lambda^1_2$ or $\lambda^2_2$ is maximal then the diffusion is limited by the single layer $G_1$ or $G_2$, and the additional information provided by the intra-layer topologies becomes relevant as it affects the diffusion process.
Note that the change between $\lambda_2$ and either $\lambda_2^1$ or $\lambda_2^2$ being maximal is related to the transitions pointed out in Ref.~\cite{ch2.5-Gomez2013,ch2.5-Radicchi2013}.

This behavior is shown in Fig.~\ref{fig:meanf} for the mean-field matrix of two interconnected networks.
The figure shows the second largest eigenvalues of $\hat{\mathbf{T}},\mathfrak{T}$ and the sparsest layer $\mathbf{T}_1$ for different amount of inter-links.
When only a few inter-links are present the interconnectivity between layers slows the process down, as it is expected.
When we increase the amount of inter-links, we can reach the convergence rate of single layers, which is the point where the single layers slow down the process.
However, with an increasing amount of inter-links the single layers lose their importance and the process is again slowed down by the inter-links.
This happens because a very large amount of inter-links force the random walker to switch between layers with increasing probability, thus, preventing diffusion to reach the whole layer.
To conclude, the mean-field transition matrix $\hat{\mathbf{T}}$ is a better estimation than $\mathfrak{T}$ in intermediate numbers of interlinks, which for our systems is in the region of approximately $550$ to $1800$ inter-links.
Otherwise, the information about the link densities as captured in $\mathfrak{T}$ is enough to approximate the second-largest eigenvalue of $\hat{\mathbf{T}}$, and thus the speed of diffusion.

\begin{figure}[t]
 \begin{center}\includegraphics[width=0.6\textwidth]{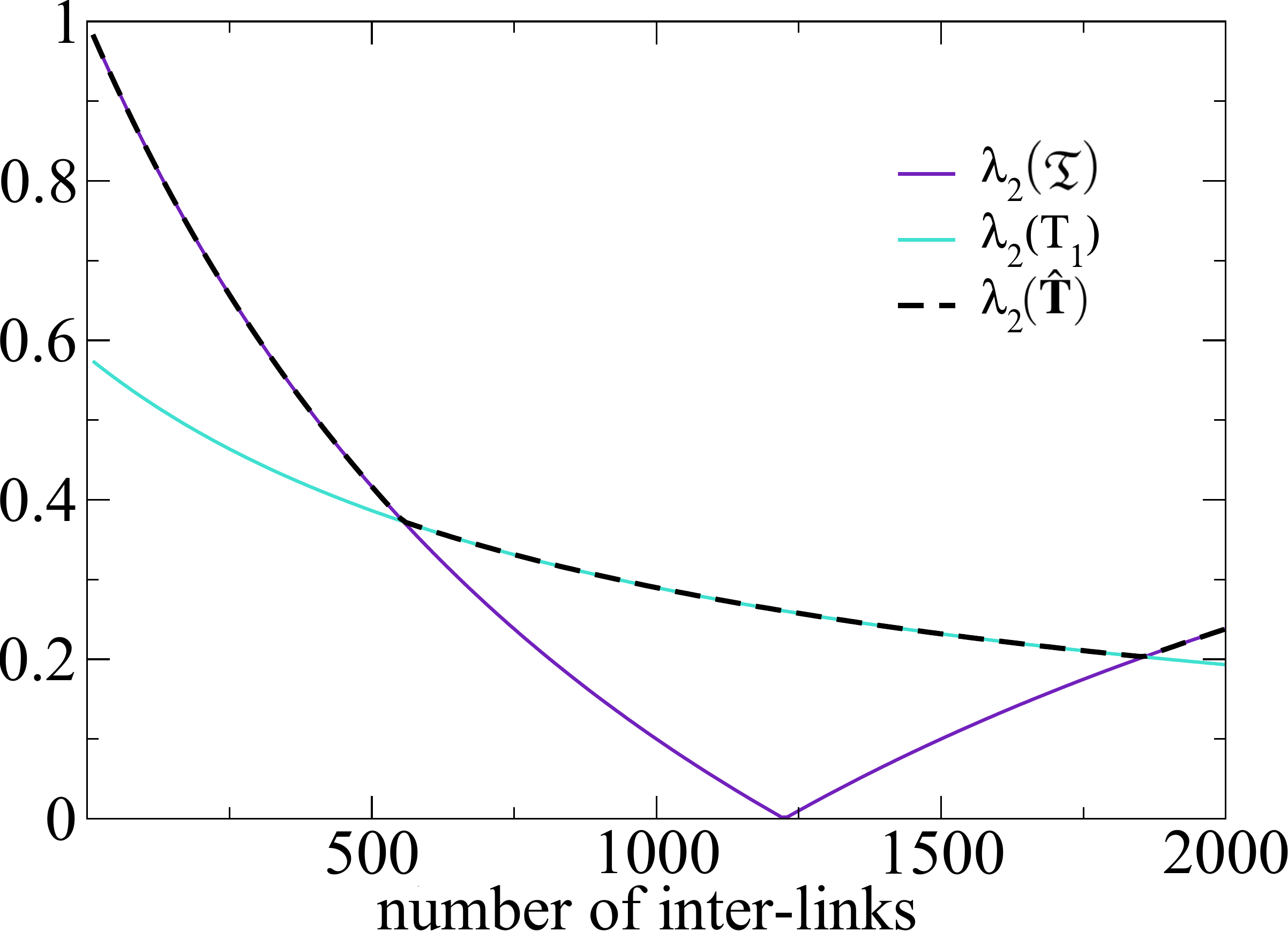}
 \caption{Eigenvalues of a mean-field approach of a two-layered network.
 Layer $1$ consists of an Erd\"os-R\'enyi network of $100$ nodes and $500$ links and Layer $2$ consists of an Erd\"os-R\'enyi network of $100$ nodes and $750$ links.
 The $x$-axis indicates the amount of inter-links randomly added across the layers.
 The lines indicate the second-largest eigenvalue of: black dashed: the mean-field supra-transition matrix $\lambda_2(\hat{\mathbf{T}})$, violet: the layer-aggregated matrix $\lambda_2(\mathfrak{T})$,
 turquoise: the larger single layer eigenvalues of $\lambda_2(\mathbf{T}_1)$ and $\lambda_2(\mathbf{T}_2)$.}
 \label{fig:meanf}
 \end{center}
\end{figure}

In general the spectrum of a mean-field matrix $\hat{\mathbf{T}}$ with unknown inter-link topology is given by
\begin{equation}
 Spec(\hat{\mathbf{T}})=\{1,\lambda_2,\ldots,\lambda_n\}\cup\left(\bigcup_{s=1}^n Spec(\mathbf{\hat{T}}_s)\setminus\lambda_1(\mathbf{\hat{T}}_s)\right)\, ,
 \label{eq:spec}
\end{equation}
or
\begin{equation}
 Spec(\hat{\mathbf{T}})=Spec(\mathfrak{T})\cup\left(\bigcup_{s=1}^n Spec(\mathbf{\hat{T}}_s)\setminus\lambda_1(\mathbf{\hat{T}}_s)\right)\, ,
 \label{eq:specag}
\end{equation}
where $\mathfrak{T}$ is the multi-layer aggregation of $\hat{\mathbf{T}}$ as described before (for details see Prop.~\ref{prop:spec} in the appendix).
This decomposition of eigenvalues  can also be useful for other network properties that depend on eigenvalues.

So far we provided an estimation based on the eigenvalues of a mean-field transition matrix $\hat{\mathbf{T}}$ that intends to approximate the ensemble average.
In reality however, ensemble realizations of multi-layer networks that contain layer $G_1$ and $G_2$ can deviate from the mean-field estimation.
This is shown in Fig.~\ref{fig:inter} (a) where we plot the second-largest eigenvalues of $\hat{\mathbf{T}}$, $\mathfrak{T}$, and ensemble averages over $100$ realizations of $\mathbf{T}$ against the number of inter-links between $G_1$ and $G_2$.
As we can see, the magenta colored dashed line showing the mean-field approximation of $\mathfrak{T}$ is a good proxy for the diffusion dynamics in the region when inter-links dominate, which is the case for either sparse or very dense inter-link topologies.
However, as shown by the cyan colored line, we can actually improve this approximation if we additionally consider the intra-links of all single layers.

There is a peak where the difference between the estimation and the ensemble averages $\Delta\lambda_2=\lambda_2(\mathbf{T})-\lambda_2(\hat{\mathbf{T}})$ reaches high values up to $0.225$, as shown in Fig.~\ref{fig:inter} (b).
This happens, on one hand, due to the large degree of freedom that comes from the absence of intra-connectivity informations within the layers.
On the other hand, the mean-field matrix assumes ``full-connectivity'' across layers, and even though this implies small weights for each single inter-link, it leads to a systematic bias towards overestimating the diffusion speed.
Nevertheless, we would like to highlight that the multi-layer aggregation provides a quite accurate estimation of the diffusion speed in the regimes where inter-links limit diffusion.

\begin{figure}[t]
  \includegraphics[width=0.46\textwidth]{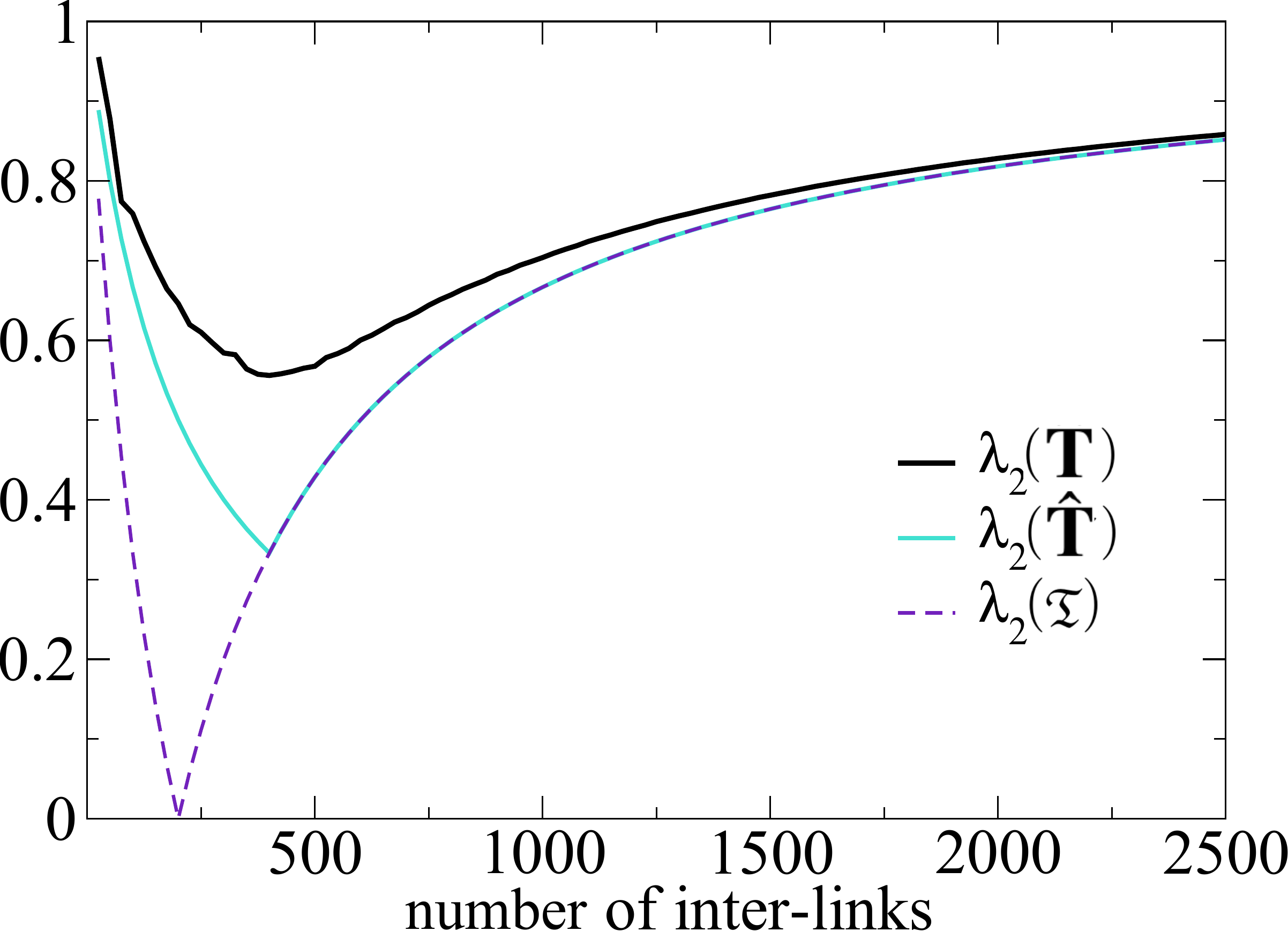}
  \includegraphics[width=0.52\textwidth]{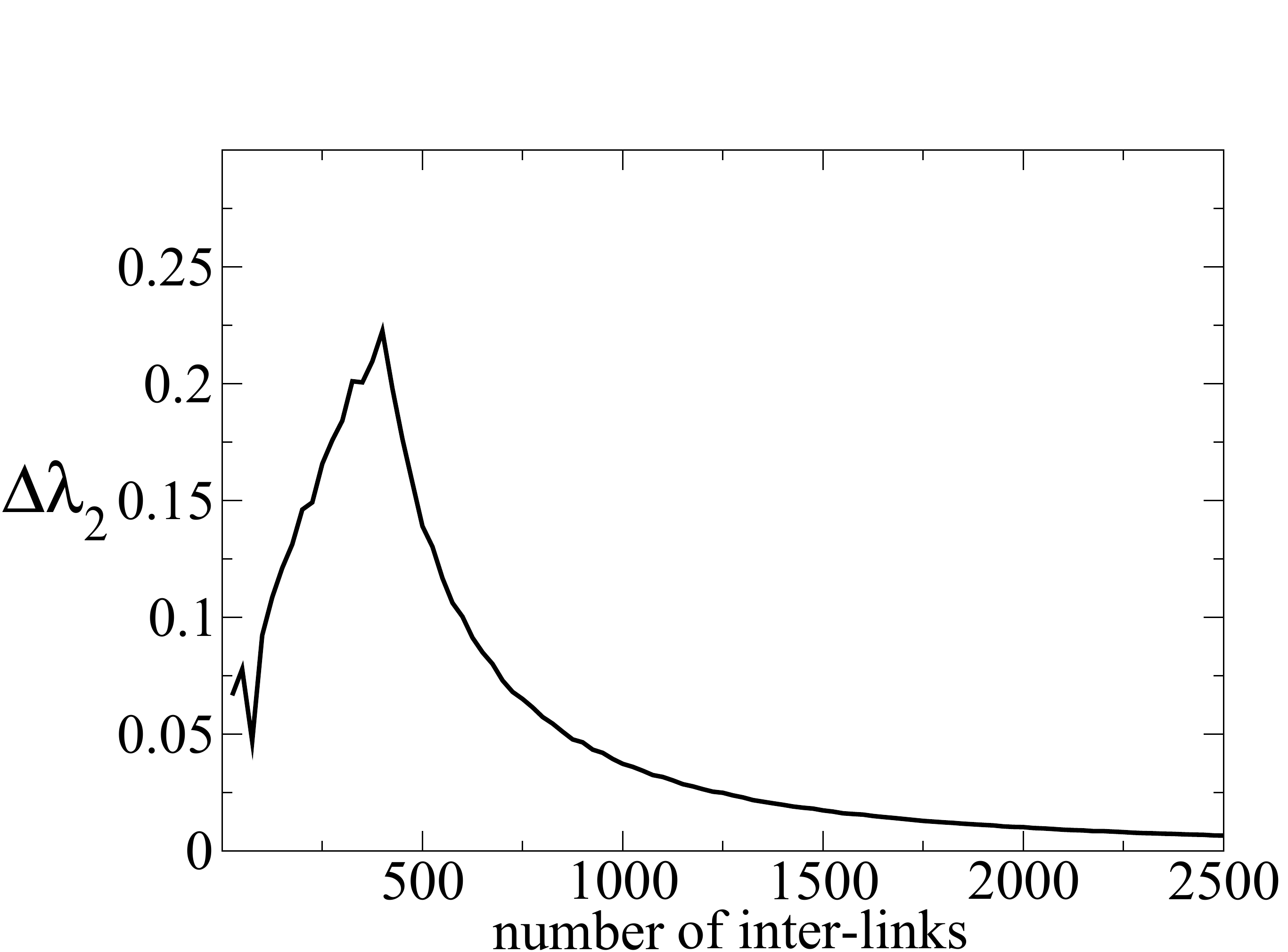}
  \begin{picture}(0,0)
  \put(-465,145){(a)}
  \put(-220,145){(b)}
  \end{picture}
  \caption{{\bf a)} Eigenvalues of a mean-field approach of a two-layered network.
  Layer $1$ and $2$ both consist of an Erd\"os-R\'enyi network of $50$ nodes and $100$ links but with different topologies.
  The $x$-axis indicates the amount of random inter-links added across layers.
  The lines indicate the second-largest eigenvalues of:
  black line: ensemble averages,
  turquoise line: mean-field estimate including intra-link topology,
  violet dashed line: mean-field only considering densities.
  {\bf b)} Eigenvalue difference between ensemble average and mean-field estimation $\Delta\lambda_2=\lambda_2(\mathbf{T})-\lambda_2(\hat{\mathbf{T}})$.
  }
  \label{fig:inter}
\end{figure}
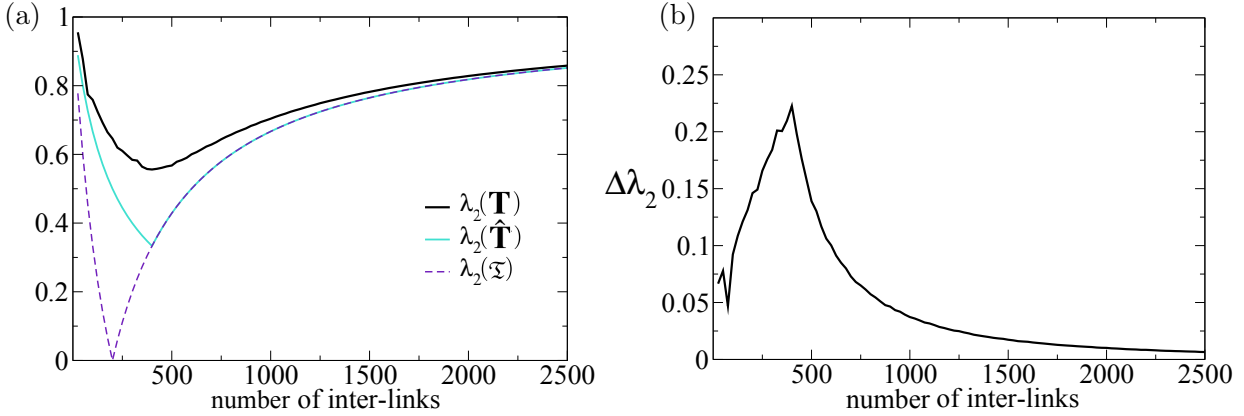

\subsection*{Case III: Unknown intra-connectivity}\label{intra}

For this case we assume full knowledge of the inter-link topology, i.e. we know exactly how the layers are connected, but the intra-link topologies, i.e. how the nodes are connected within the single layers, are unknown.
More precisely, we consider two interconnected layers $G_1$ and $G_2$ of a multi-layer network, and we fix the inter-links $E_I(\mathbf{G})$ in a bipartite network structure that connects nodes of $G_1$ to nodes of $G_2$.
Since we have no information about the intra-link topologies of $G_1$ and $G_2$, we assume random connectivities within the layers, so that we only know the average degrees $\hat{d}_1$ and $\hat{d}_2$ of $G_1$ and $G_2$ respectively.
This means that the off-diagonal blocks $\mathbf{A}_{12}^\top=\mathbf{A}_{21}$ of the supra-adjacency matrix are given, but the diagonal blocks $\mathbf{A}_1$ and $\mathbf{A}_2$ are unknown.

Because we only know the average degrees $\hat{d_1}$ and $\hat{d_2}$ of the layers, we can define mean-field versions of the adjacency matrices such that
\begin{equation*}
 \hat{\mathbf{A}_1}=\left\{\frac{\hat{d_1}}{n_1}\right\}_{ij}\quad\text{and}\quad \hat{\mathbf{A}_2}=\left\{\frac{\hat{d_2}}{n_2}\right\}_{ij}\,.
\end{equation*}
However, even though we know the topology of the inter-links, we do not know which nodes exactly are connected to each other.
Hence we use the same approach as in Case II with $m$ equal to the amount of inter-links and the factors defined as in Eq.(\ref{eq:fac}).
Therefore we get the mean-field transition matrix $\mathbf{\mathbf{\hat{T}}}$ consisting of the following block matrices,
\begin{align}
 \mathbf{\hat{T}}_{1}&=\left\{\frac{\hat{d_1}}{n_1\hat{d_1}+m}\right\}_{ij}\quad\text{for}\quad i\in\{1,\ldots,n_1\},j\in\{1,\ldots,n_2\}\\
 \mathbf{\hat{T}}_{2}&=\left\{\frac{\hat{d_2}}{n_2\hat{d_1}+m}\right\}_{ij}\quad\text{for}\quad i\in\{1,\ldots,n_2\},j\in\{1,\ldots,n_1\}\, .
\end{align}
The off-diagonal blocks are just rescaled transition matrices of $\mathbf{A}_{12}$ and $\mathbf{A}_{21}$,
\begin{equation}
 \mathbf{\hat{T}}_{12}=\alpha_{12}T(\mathbf{A}_{12}),\quad \mathbf{\hat{T}}_{21}=\alpha_{21}T(\mathbf{A}_{21})\,.
\end{equation}

However, this time we are not able to compute exactly the single layer eigenvalues $\lambda^1_1$ and $\lambda^2_2$, as it was the case in Case II.
In particular, depending on the ensemble constraints we could only compute an average eigenvalue $\hat{\lambda_2}$ for a single layer.
Therefore, we can use the following maximization term
\begin{equation}
 \lambda_2(\hat{\mathbf{T}})=\max\left(1-\alpha_{12}-\alpha_{21},(1-\alpha_{12})\hat{\lambda^1_2},(1-\alpha_{21})\hat{\lambda^2_2}\right)\,,
\end{equation}
which is has the same form as in Case II (see Eq.(\ref{eq:maxterm1})).
Here again, as long as $\lambda_2=1-\alpha_{12}-\alpha_{21}$ is maximal the inter-links are the limiting factor of diffusion in the multi-layer network, which means that due to the inter-link topology the random walk diffusion is slowed down, and the influence of the intra-layer topologies is marginal to the process.
On the other hand, when either the average term of $\hat{\lambda^1_2}$ or $\hat{\lambda^2_2}$ is maximal then the diffusion is limited by the single layer $G_1$ or $G_2$, and the additional information provided by the intra-layer topologies becomes relevant as it affects the diffusion process.

In Fig.~\ref{fig:intra}(a), starting with initially empty intra-networks\footnote{Note that even though the intra-layer networks are empty initially, there is a number of inter-layer links which provide connectivity across the layers, similar to a bipartite network.}, we plot the second largest eigenvalues of $\mathfrak{T}$, $\hat{\mathbf{T}}$, and the ensemble average of $100$ realizations of $\mathbf{T}$ against the number of intra-links that are simultaneously and randomly added in both layers.
We observe that the general behavior is similar to Fig.~\ref{fig:inter}.
Thus, the multi-layer aggregation plotted in magenta approximates well the regions where the inter-links are the relevant factor, which is for very sparse and increasingly dense intra-links densities.
The difference between the mean-field and the ensemble average $\Delta\lambda_2=\lambda_2(\mathbf{T})-\lambda_2(\hat{\mathbf{T}})$ as seen in Figure~\ref{fig:intra}(b) again rises up to a peak of about $0.225$.

Our analysis shows that there is some form of symmetry in knowing the degree of the nodes in the single layers, but not knowing how they are connected to nodes in other layers and to knowing the inter-links between layers, but not the degree of their adjacent nodes.
Even though the ensembles generated from these two constraints can be much different, the relevance of inter-links or intra-links of all single layers to a diffusive process is comparable for both cases.

\begin{figure}[t]
  \includegraphics[width=0.46\textwidth]{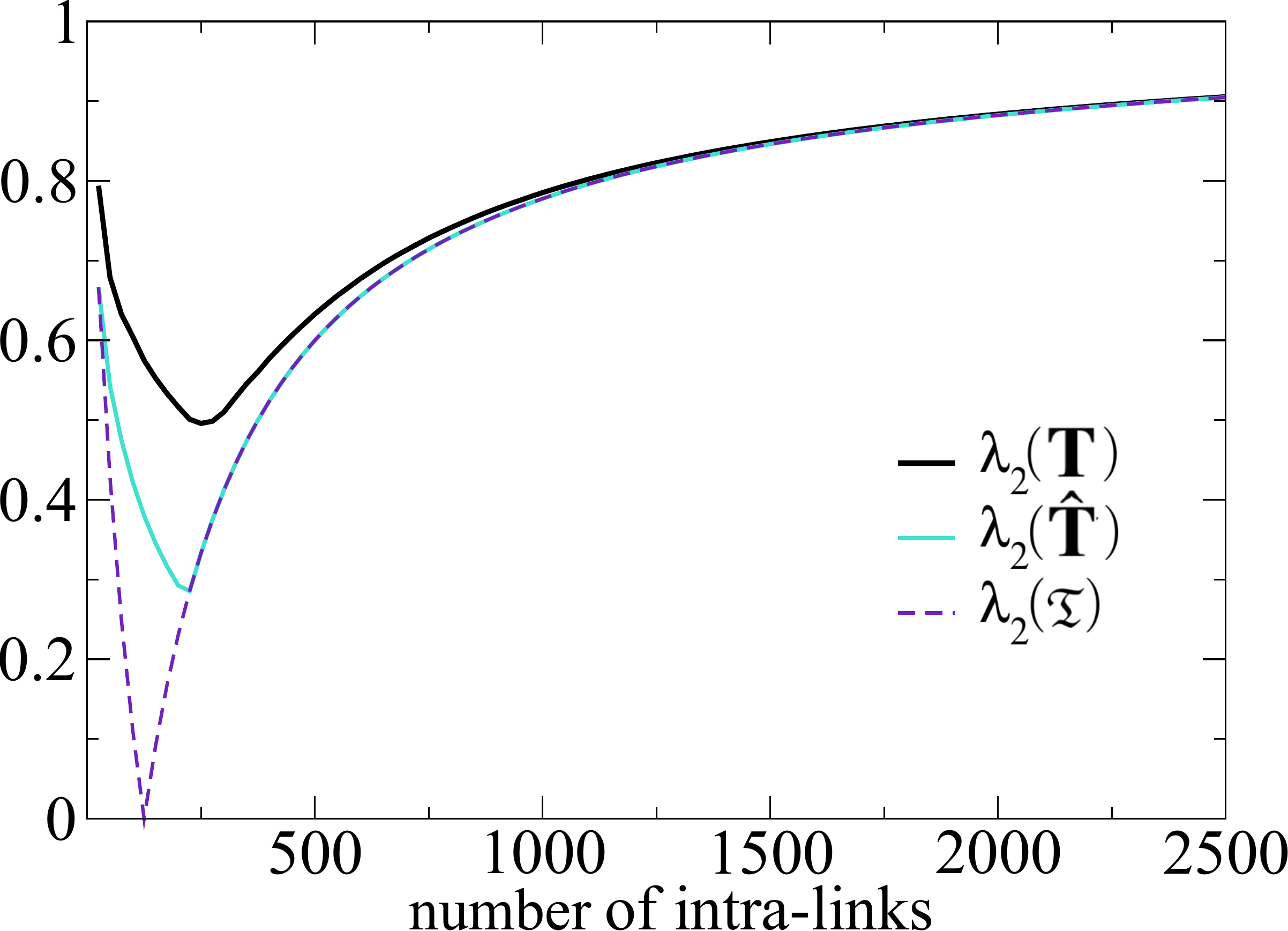}
  \includegraphics[width=0.52\textwidth]{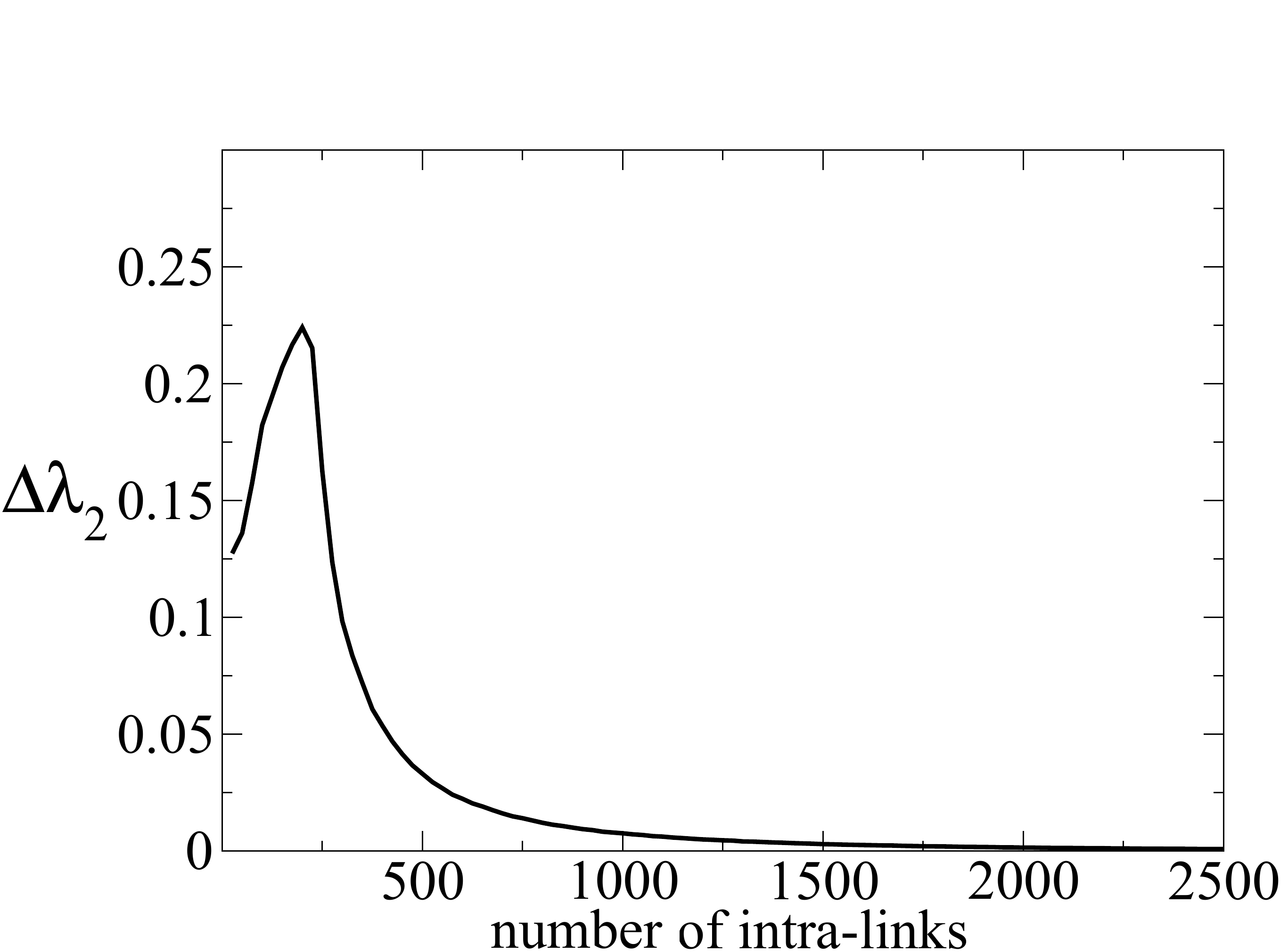}
  \begin{picture}(0,0)
  \put(-465,145){(a)}
  \put(-220,145){(b)}
  \end{picture}
 \caption{{\bf a)} Eigenvalues of a mean-field approach of a tow-layered network with $250$ inter-links.
 Layer $1$ and $2$  both consist of $50$ nodes but no edges.
 The $x$-axis indicates the amount of links intra-links that are simultaneously added to both layers.
 The lines indicate the second-largest eigenvalue of:
 black line: ensemble averages,
 turquoise line: mean-field estimate including inter-links,
 dashed violet line: mean-field only considering link densities.
 {\bf b)} Eigenvalue Difference between ensemble average and mean-field approach $\Delta\lambda_2=\lambda_2(\mathbf{T})-\lambda_2(\hat{\mathbf{T}})$.}
 \label{fig:intra}
\end{figure}
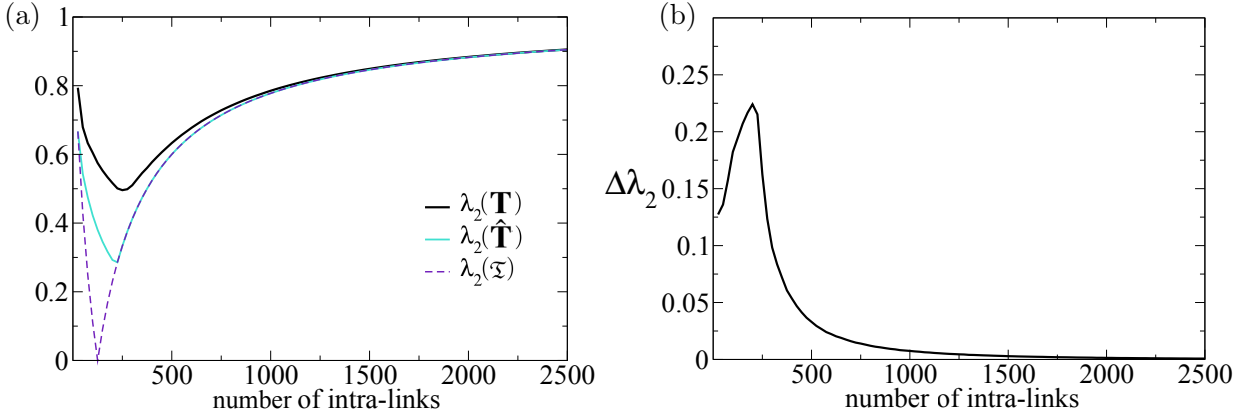

\section{Conclusion}\label{conclusion}

In this manuscript, we showed how an ensemble perspective can be applied to multi-layer networks in order to address realistic scenarios when only limited information is available.
More precisely, we focused on a diffusion process that runs on the multi-layer network and its relation to the spectrum of the supra-transition matrix.
We have shown that the convergence rate of the diffusion process is limited by either the inter-links or intra-links of the single layers and we identified for which relation of inter-link compared to intra-link densities it is sufficient to only consider transitions across layers, instead of the full information on all individual nodes.
This implies that we do not always need perfect information to make statements about a multi-layer network because, under certain conditions, we are still able to make analytical statements about the network only using partial information.
In realistic situations data can be an issue either due to availability constraints or due to their vast amounts.
In such cases, even though an exact analysis is impossible, we may still derive useful conclusions about processes that depend on the network spectrum (like diffusion and synchronization) using only aggregated statistics.

For our study we assumed the simplest case of random networks, therefore exploring other ways to couple the network layers or including link-weights and directed links and testing their influence on our results is up to future investigation.

\paragraph*{Acknowledgements.}
N.W., A.G. and F.S. acknowledge support from the EU-FET project MULTIPLEX 317532.

\section{Appendix}
\textbf{Note}: Unless stated otherwise, here vectors are considered to be row-vectors and multiplication of vectors with matrices are left multiplications.
\vspace{6pt}

We assume a multi-layer network $\mathbf{G}$ consisting of $L$ layers $G_1,\ldots,G_L$ and $n$ nodes.
A single layer $G_s$ contains $n_s$ nodes and therefore $\sum_{s=1}^L n_s=n$.
For a multi-layer network $\mathbf{G}$ we define the \textit{supra-transition} matrix that can be represented in block structure according to the layers:
\begin{equation*}
  \mathbf{T} =
  \left(
    \begin{array}{c|c|c|c|c}
      \mathbf{T}_1 & \ldots & \mathbf{T}_{1t}& \ldots& \mathbf{T}_{sL} \\
      \hline
      \vdots & \ddots & \vdots& \ddots& \vdots \\
      \hline
      \mathbf{T}_{s1} & \ldots & \mathbf{T}_{st}&\ldots&  \mathbf{T}_{sL} \\
      \hline
      \vdots & \ddots & \vdots& \ddots& \vdots \\
      \hline
      \mathbf{T}_{L1} & \ldots & \mathbf{T}_{Lt}&\ldots&  \mathbf{T}_L
    \end{array}
   \right)\,.
\end{equation*}
Each $\mathbf{T}_{st}$ contains all the transition probabilities from nodes in $G_s$ to nodes in $G_t$.
Assuming Eq.(\ref{eqn:cond})
it follows that $\mathbf{T}_{st}=\alpha_{st}\mathbf{R}_{st}$ where $\mathbf{R}_{st}$ is a row stochastic matrix.
This means that all $\mathbf{T}_{st}$ are scaled transition matrices.
The factor $\alpha_{st}$ represents the weighted fraction of all links starting in $G_s$ that end up in $G_t$.
\\
In this respect we define the aggregated transition matrix $\mathfrak{T}$ of dimension $L$,
\begin{equation}
 \mathfrak{T}=
  \left(
    \begin{array}{c|c|c|c|c}
      \alpha_{11} & \ldots & \alpha_{1t}& \ldots& \alpha_{sL} \\
      \hline
      \vdots & \ddots & \vdots& \ddots& \vdots \\
      \hline
      \alpha_{s1} & \ldots & \alpha_{st}&\ldots&  \alpha_{sL} \\
      \hline
      \vdots & \ddots & \vdots& \ddots& \vdots \\
      \hline
      \alpha_{L1} & \ldots & \alpha_{Lt}&\ldots&  \alpha_L
    \end{array}
   \right)\,.
\end{equation}
Each vector $v$ of dimension $n$ can be split according to the layer-separation given by $\mathbf{G}$,
\begin{equation*}
 v=\left(v^{(1)},\ldots,v^{(k)},\ldots,v^{(L)}\right)\,.
\end{equation*}
Each component $v^{(k)}$ has exactly dimension $n_k$.
We define the \textit{layer-aggregated vector} $\mathfrak{v}=(\mathfrak{v}_1,\ldots,\mathfrak{v}_L)$ of dimension $L$ as follows
\begin{equation*}
 \forall k\in\{1,\ldots,L\}\quad \mathfrak{v}_k=\sum_{i=1}^{n_k} \left[v^{(k)}\right]_i\,.
\end{equation*}
We use the bracket notation $[v]_i$ to represent the $i$-th entry of the vector $v$.
Analogously, by $[v\mathbf{M}]_i$ we mean the $i$-th entry $w_i$ that represents the multiplication of $v$ with a matrix $\mathbf{M}$, i.e. $w=v\mathbf{M}$. 
Further, by $|v|$ we indicate the sum of the entries of $v$, $|v|=\sum_i v_i=\sum_i [v]_i$.

\begin{thm}\label{thm:main}
For a multi-layer network $\mathbf{G}$ consisting of $L$ layers we assume the supra-transition matrix $\mathbf{T}$ to consist of block matrices $\mathbf{T}_{st}$ such that
for all $s,t\in\{1,\ldots,L\}$, $\mathbf{T}_{st}=\alpha_{st}\mathbf{R}_{st}$ where $\alpha_{st}\in\mathbb{Q}$ and $\mathbf{R}_{st}$ is a row stochastic transition matrix.
The multi-layer aggregation is defined by ${\mathfrak{T}}=\{\alpha_{st}\}_{st}$.
If an eigenvalue $\lambda$ of the matrix $\mathbf{T}$ corresponds to an eigenvector $v$ with a layer-aggregation $\mathfrak{v}$ that satisfies $|\mathfrak{v}|\neq 0$
then $\lambda$ is also an eigenvalue of ${\mathfrak{T}}$.
\end{thm}
\begin{proof}
 Assume $v$ is a left eigenvector of $\mathbf{T}$ corresponding to the eigenvalue $\lambda$.
 Therefore, it holds that $v\mathbf{T}=\lambda v$.
 Let $v^{(k)}$ be the $k$-th part of $v$ that corresponds to the layer $G_k$.
 We can write the matrix multiplication in block structure.
 \begin{equation*}
  \left(v^{(1)},\ldots,v^{(k)},\ldots,v^{(L)}\right)\mathbf{T}=\left(\sum_l v^{(l)}\mathbf{T}_{l1},\ldots,\sum_l v^{(l)}\mathbf{T}_{lk},\ldots,\sum_l v^{(l)}\mathbf{T}_{lL}\right)\,.
 \end{equation*}
 Each $v^{(k)}$ is a row vector which length is equal to the amount of nodes $n_k$ in $G_k$.
 The transformation $\sum_l v^{(l)}\mathbf{T}_{lk}$ is also a row vector with the same length as $v^{(k)}$.
 According to the eigenvalue equation it holds that for all $k\in\{1,\ldots,L\}$
 \begin{equation*}
  \lambda v^{(k)}=\left(v\mathbf{T}\right)^{(k)}=\sum_l v^{(l)}\mathbf{T}_{lk}\,.
 \end{equation*}
 Now let us denote the sum of the vector entries of $v^{(k)}$ as
 \begin{equation*}
   \mathfrak{v}_k=\sum_i \left[v^{(k)}\right]_i\,.
 \end{equation*}
 Further, we define layer-aggregated vector consisting of this sums by $\mathfrak{v}=\left(\mathfrak{v}_1,\ldots,\mathfrak{v}_n\right)$.
 Note that for a general row stochastic matrix $\mathbf{M}$ and its multiplication with an eigenvector $v$ to the eigenvalue $\lambda$ it holds that $\lambda\sum_j [v]_j=\sum_j [v\mathbf{M}]_j$.
 For the components after multiplication with $\mathbf{T}$ we can deduce
 \begin{align*}
  \sum_i\left[\left(v\mathbf{T}\right)^{(k)}\right]_i&=\sum_i\left[\sum_l v^{(l)}\mathbf{T}_{lk}\right]_i=\sum_i\left[\sum_l \alpha_{lk}v^{(l)}\mathbf{R}_{lk}\right]_i\\
  &=\sum_l \alpha_{lk}\sum_i\left[v^{(l)}\mathbf{R}_{lk}\right]_i=\sum_l \alpha_{lk}\sum_i\left[v^{(l)}\right]_i=\sum_l \alpha_{lk}\mathfrak{v}_l\,.
 \end{align*}
 If we multiply $\mathfrak{v}$ with $\mathfrak{T}$ and look at a single entry of $\mathfrak{v}\mathfrak{T}$ we get
 \begin{equation*}
  \left[\mathfrak{v}\mathfrak{T}\right]_k=\sum_l\mathfrak{v}_l\mathfrak{T}_{lk}=\sum_l\mathfrak{v}_l\alpha_{lk}\,.
 \end{equation*}
 Hence it holds that
 \begin{equation*}
  \left[\mathfrak{v}\mathfrak{T}\right]_k=\sum_i\left[\left(v\mathbf{T}\right)^{(k)}\right]_i\, ,
 \end{equation*}
 and therefore
 \begin{equation*}
  \mathfrak{v}\mathfrak{T}=\left(\sum_i\left[\left(v_2\mathbf{T}\right)^{(1)}\right]_i,\ldots,\sum_i\left[\left(v_2\mathbf{T}\right)^{(L)}\right]_i\right)\,.
 \end{equation*}
 Finally since $\mathbf{T}$ is row stochastic and $\lambda v=v\mathbf{T}$ we have that
 \begin{align*}
  \lambda \mathfrak{v}&=\lambda\left(\mathfrak{v}_1,\ldots,\mathfrak{v}_n\right)\\
  &=\lambda\left(\sum_i[v^{(1)}]_i,\ldots,\sum_i[v^{(L)}]_i\right)\\
  &=\left(\sum_i\left[\left(v\mathbf{T}\right)^{(1)}\right]_i,\ldots,\sum_i\left[\left(v\mathbf{T}\right)^{(L)}\right]_i\right)\\
  &=\mathfrak{v}\mathfrak{T}\,.
 \end{align*}
 Therefore, $\lambda$ is also an eigenvalue of $\mathfrak{T}$ to the eigenvector $\mathfrak{v}$ defined as before.
 It is important to note that this only holds if $|\mathfrak{v}|\neq 0$.
\end{proof}

The procedure used in the proof of the previous theorem applies to several eigenvalues of $\mathbf{T}$ but at most $L$ of them.
Next we give a proposition for the reversed statement of Thm~\ref{thm:main}.

\begin{prop}\label{prop:rev}
Let $\mathbf{G}$ be a multi-layer network that consists of $L$ layers and fulfills all of the conditions of Thm~\ref{thm:main}.
Let $\mathfrak{T}=\{\alpha_{st}\}_{st}$ be the multi-layer aggregation of $\mathbf{T}$.
If $\lambda$ is an eigenvalue of $\mathfrak{T}$ then $\lambda$ is also an eigenvalue of $\mathbf{T}$.
\end{prop}
\begin{proof}
 Assume that $\lambda$ is an eigenvalue of $\mathfrak{T}$.
 For each matrix there exist a left and right eigenvector that correspond to the same eigenvalue $\lambda$.
 Assume the $\mathfrak{w}$ is the right eigenvector and therefore a column vector.
 Hence $\mathfrak{T}\mathfrak{w}=\lambda\mathfrak{w}$ and
 \begin{equation}
  \mathfrak{T}\mathfrak{w}=\left(\sum_j \alpha_{1j}\mathfrak{w}_j,\ldots,\sum_j \alpha_{kj}\mathfrak{w}_j,\ldots,\sum_j \alpha_{Lj}\mathfrak{w}_j\right)^\top=\lambda\mathfrak{w}\,.
 \end{equation}
 Now we generate a column vector $w$ of dimension $n$ such that for all the layer components $w^{(k)}$ it holds that
 \begin{equation}
  w^{(k)}=(\mathfrak{w}_k,\ldots,\mathfrak{w}_k)^\top\,.
 \end{equation}
 Next we perform a right multiplication of $w$ with $\mathbf{T}$,
  \begin{align*}
  \mathbf{T}w&=\left(\sum_{l} \mathbf{T}_{1l}w^{(l)},\ldots,\sum_{l} \mathbf{T}_{kl}w^{(l)},\ldots,\sum_{l} \mathbf{T}_{Ll}w^{(l)}\right)^\top\\
  &=\left(\sum_{l} \alpha_{1l}\mathbf{R}_{1l}w^{(l)},\ldots,\sum_{l} \alpha_{kl}\mathbf{R}_{kl}w^{(l)},\ldots,\sum_{l} \alpha_{Ll}\mathbf{R}_{Ll}w^{(l)}\right)^\top\,.
  \end{align*}
  Since all $\mathbf{R}_{st}$ are row stochastic matrices and $w^{(l)}$ contains only the value $\mathfrak{w}_l$ for each entry we get $\mathbf{R}_{st}w^{(l)}=w^{(l)}$.
  It follows that
    \begin{align*}
  \mathbf{T}w&=\left(\sum_{l} \alpha_{1l}w^{(l)},\ldots,\sum_{l} \alpha_{kl}w^{(l)},\ldots,\sum_{l} \alpha_{Ll}w^{(l)}\right)^\top\\
  &=\left(\lambda w^{(1)},\ldots,\lambda w^{(k)},\ldots,\lambda w^{(L)}\right)^\top\\
  &=\lambda w\,.
  \end{align*}
  And therefore $\lambda$ is also an eigenvalue of $\mathbf{T}$.
\end{proof}

In the case of a diffusion process we are especially interested in the second-largest eigenvalue of $\mathbf{T}$, denoted by $\lambda_2(\mathbf{T})$, which is related to algebraic connectivity of $\mathbf{T}$.
In this perspective the following corollary is useful:
\begin{cor}\label{cor:eig}
Let $\mathbf{G}$ be a multi-layer network consisting of $L$ layers that fulfill all of the conditions of Thm~\ref{thm:main}.
Further assume that $\mathbf{G}$ is partitioned according to a spectral partitioning, i.e. according to the eigenvector corresponding to $\lambda_2(\mathbf{T})$, then $\lambda_2(\mathbf{T})=\lambda_2(\mathfrak{T})$.
\end{cor}
\begin{proof}
 In general all the eigenvectors of a transition matrix, except the eigenvector corresponding to the largest eigenvalue that is equal to one, sum up to zero.
 However, these eigenvectors consist of positive and negative entries that allow for a spectral partitioning.
 Especially the eigenvector $v_2$ that corresponds to the second-largest eigenvalue $\lambda_2(\mathbf{T})$, can be used for the partitioning of the network.
 This eigenvector is related to the Fiedler vector that is also used for spectral bisection~\cite{ch2.5-Fiedler1973}.
 Therefore if the layer-partition of $\mathbf{G}$ coincides with this spectral partitioning we assure that the layer-aggregated vector of $v_2$  satisfies $|\mathfrak{v}_2|\neq 0$.
 Considering this and Prop~\ref{prop:rev} the corollary follows directly from Thm~\ref{thm:main}.
\end{proof}

Given Eq.(\ref{eqn:cond}) we can fully describe the spectrum of $\mathbf{T}$ based on the intra-layers transition blocks $\mathbf{T}_{i}$ for $i\in\{1,\ldots,n\}$ and the spectrum of $\mathfrak{T}$.
Note that with uniform columns of a matrix $\mathbf{M}$ we mean that each column of $\mathbf{M}$ contains the same value in each entry.
However, this value can be different for different columns.
\begin{prop}\label{prop:spec}
Let $\mathbf{T}$ be the supra-transition matrix of a multi-layer network $\mathbf{G}$ that consist of $L$ layers and satisfies Eq.(\ref{eqn:cond}).
If $\mathbf{T}$ has off-diagonal block matrices $\mathbf{T}_{st}$, for $s,t\in\{1,\ldots,n\}$ and $s\neq t$, that all have uniform columns, then the spectrum of $\mathbf{T}$ can be decomposed as
\begin{equation}
 Spec(\mathbf{T})=\{1,\lambda_2,\ldots,\lambda_L\}\cup\left(\bigcup_{s=1}^L Spec(\mathbf{T}_s)\setminus\{\lambda_1(\mathbf{T}_s)\}\right)\, ,
\end{equation}
where $\mathbf{T}_s$ are the block matrices of $\mathbf{T}$ corresponding to the single layers $G_s$ and $\lambda_1(\mathbf{T}_s)$ the largest eigenvalue of $\mathbf{T}_s$.
The eigenvalues $\lambda_2,\ldots,\lambda_L$ are attributed to the interconnectivity of layers.
\end{prop}
\begin{proof}
 To prove this statement we just have to show that all eigenvalues (except the largest one) of $\mathbf{T}_s$ for $s\in\{1,\ldots,L\}$ are also eigenvalues of $\mathbf{T}$.
 We assume that $\lambda$ is any eigenvalue corresponding to the eigenvector $u$ of some block matrix $\mathbf{T}_r$, i.e. $\lambda u=u\mathbf{T}_r$.
 We define a row vector $v$ that is zero everywhere except at the position where it corresponds to $\mathbf{T}_r$.
 The vector $v$ looks like $v=(0,\ldots,0,u,0,\ldots,0)$.
 Now we investigate what happens if we multiply this vector with the transition matrix $\mathbf{T}$.
 \begin{equation*}
  v\mathbf{T}=\left(v^{(1)},\ldots,v^{(k)},\ldots,v^{(n)}\right)\mathbf{T}=\left(\sum_l v^{(l)}\mathbf{T}_{l1},\ldots,\sum_l v^{(l)}\mathbf{T}_{lk},\ldots,\sum_l v^{(l)}\mathbf{T}_{ln}\right)\,.
 \end{equation*}
 Let us take a look at the effect of the matrix multiplication on an arbitrary component $v^{(k)}$ with $k\neq r$ and recall that $v^{(k)}$ is equal to a zero vector $\mathbf{0}$ for $k\neq r$.
 \begin{equation*}
  (v\mathbf{T})^{(k)}=\sum_l v^{(l)}\mathbf{T}_{lk}=\sum_{l,l\neq r} \mathbf{0}\mathbf{T}_{lk}+u\mathbf{T}_{rk}=u\mathbf{T}_{rk}\,.
 \end{equation*}
 Note that all eigenvectors $u$ of a transition matrix that are not related to the largest eigenvalue sum up to zero.
 Therefore it holds that $u\mathbf{T}_{rk}=\mathbf{0}$ since $\mathbf{T}_{rk}$ has uniform columns and therefore $u\mathbf{T}_{rk}$ yields in a vector where each entry is equal to some multiple of $|u|$.
 In case of $k=r$ it holds that $v^{(k)}=u$ and we get
 \begin{equation*}
  (v\mathbf{T})^{(r)}=\sum_l v^{(l)}\mathbf{T}_{lr}=\sum_{l,l\neq r} \mathbf{0}\mathbf{T}_{lk}+u\mathbf{T}_{rr}=u\mathbf{T}_{r}=\lambda r\,.
 \end{equation*}
 Hence, it holds that $v\mathbf{T}=\lambda v$, which means that $\lambda$ is also an eigenvalue of $\mathbf{T}$.
 This way we get $n-L$ eigenvalues of $\mathbf{T}$ apart from the largest eigenvalue that is equal to one.
 The remaining eigenvalues denoted by $\lambda_2,\ldots,\lambda_L$ are not attributed to any block matrix of $\mathbf{T}$.
 Therefore they are considered to be the interconnectivity eigenvalues.
\end{proof}

\begin{cor}\label{cor:agg}
Let $\mathbf{G}$ be a multi-layer network consisting of $L$ layers that satisfies Eq.(\ref{eqn:cond}) and the conditions of Prop~\ref{prop:spec}.
Then the aggregated matrix $\mathfrak{T}=\{\alpha_{st}\}_{st}$ has spectrum
\begin{equation*}
 Spec(\mathfrak{T})=\{1,\lambda_2,\ldots,\lambda_L\}\, ,
\end{equation*}
and it holds that $\lambda_2,\ldots,\lambda_L\in Spec(\mathbf{T})$.
\end{cor}
\begin{proof}
 Note that every eigenvalue $\lambda\neq 1$ of some block matrix $\mathbf{T}_r$ with $\lambda u=u\mathbf{T}_r$ is by Prop~\ref{prop:spec} also an eigenvalue of $\mathbf{T}$.
 Furthermore, $\lambda$ is attributed to the eigenvector $v=(0,\ldots,0,u,0,\ldots,0)$ of $\mathbf{T}$.
 However $|v|=0$ since $u$ is an eigenvector of a transition matrix, not related to the largest eigenvalue, and therefore sums up to zero.
 Hence all eigenvalues fulfilling this condition are by Thm~\ref{thm:main} not eigenvalues of $\mathfrak{T}$.
 Since $\mathfrak{T}$ contains at least $L$ eigenvalues that by Prop~\ref{prop:rev} also correspond to eigenvalues of $\mathbf{T}$,
 the remaining eigenvalues $\lambda_2,\ldots,\lambda_L$ have to also be eigenvalues of $\mathfrak{T}$.
\end{proof}

In the following we provide a useful proposition for the eigenvalues arising from the inter-links in case of two layers.
Note that by the function $T(\cdot)$ applied to matrix $\mathbf{M}$ we indicate that $T(\mathbf{M})$ is the row-normalization of $\mathbf{M}$.
\begin{prop}\label{prop:bip}
Let $\mathbf{G}$ be a multi-layer network that satisfies Eq.(\ref{eqn:cond}), consisting of two networks $G_1$ and $G_2$ in separate layers.
Assume that the supra-transition matrix $\mathbf{T}$ has the form
\begin{equation*}
 \mathbf{T}=
  \left(
    \begin{array}{c|c}
      \mathbf{T}_{1} & \mathbf{T}_{12}\\
      \hline
      \mathbf{T}_{21} & \mathbf{T}_{2}
    \end{array}
   \right)=
   \left(
    \begin{array}{c|c}
      \beta T(\mathbf{A}_1) & (1-\beta)\mathbf{T}^I_{12}\\
      \hline
      (1-\beta)\mathbf{T}^I_{21} & \beta T(\mathbf{A}_2)
    \end{array}
   \right)
   \, ,
\end{equation*}
where $\mathbf{T}^I$ is the transition matrix of the layer $\mathbf{G}$ that only consists of the inter-layer links and $\beta\in\mathbb{Q}$ is a constant.
Furthermore, assume that $\mathbf{T}_1$ and $\mathbf{T}_2$ have uniform columns.
\\
Then for $\lambda\in Spec(\mathbf{T}^I)$ and $\lambda\neq 1,-1$ it holds that $(1-\beta)\lambda\in Spec(\mathbf{T})$.
\end{prop}
\begin{proof}
 If $v$ is an eigenvector to the eigenvalue $\lambda\neq 1,-1$ of $\mathbf{T}^I$ it holds that $v\mathbf{T}^I=\lambda v$.
 Hence,
 \begin{equation*}
  v\mathbf{T}^I=\left(v^{(1)},v^{(2)}\right)\mathbf{T}^I=\left(v^{(2)}\mathbf{T}^I_{21},v^{(1)}\mathbf{T}^I_{12},\right)=\lambda\left(v^{(1)},v^{(2)}\right)\,.
 \end{equation*}
 Because $\lambda v^{(2)}=v^{(1)}\mathbf{T}^I_{12}$, we get $\lambda^2 v^{(1)}=v^{(1)} \mathbf{T}^I_{12}\mathbf{T}^I_{21}$.
 Therefore, $v^{(1)}$ is also an eigenvector of the transition matrix $\mathbf{T}^I_{12}\mathbf{T}^I_{21}$ to the eigenvalue $\lambda^2$.
 Note that $\lambda\neq 1,-1$ hence $\lambda^2<1$ which implies that $v^{(1)}$ does not correspond to the largest eigenvalue and therefore its entries sum up to zero.
 The same holds for $v^{(2)}$ and the matrix $\mathbf{T}^I_{21}\mathbf{T}^I_{12}$.
 For the multiplication of $v$ with $\mathbf{T}$ we deduce that
  \begin{equation*}
  v\mathbf{T}=\left(v^{(1)},v^{(2)}\right)\mathbf{T}=\left(v^{(1)}\mathbf{T}_1+(1-\beta)v^{(2)}\mathbf{T}^I_{21},(1-\beta)v^{(1)}\mathbf{T}^I_{12}+v^{(2)}\mathbf{T}_2,\right)\,.
 \end{equation*}
 Since $\mathbf{T}_1$ and $\mathbf{T}_2$ have uniform columns we get $v^{(1)}\mathbf{T}_1=\mathbf{0}$ and $v^{(2)}\mathbf{T}_2=\mathbf{0}$.
 And therefore $v\mathbf{T}=(1-\beta)\lambda v$ and $(1-\beta)\lambda\in Spec(\mathbf{T})$.
\end{proof}
Proposition~\ref{prop:bip} can be extended to multiple layers, however the proof is more involved and will be omitted.

\end{document}